\documentclass[times, 10pt,twocolumn]{article} 
\usepackage{latex8}
\usepackage{times}

\usepackage{theorem}
\usepackage{amsmath}
\usepackage{amssymb}
\usepackage{verbatim}
\usepackage{times}
\usepackage{fullpage}

\newtheorem{theorem}                          {Theorem}[section]
\newtheorem{lemma}         [theorem] {Lemma}

\newtheorem{claim}         [theorem] {Claim}
\newtheorem{corollary}         [theorem] {Corollary}

{\theorembodyfont{\rmfamily} 
\newtheorem{definition} [theorem]{Definition} {\theorembodyfont{\rmfamily} 

\newenvironment{proof}{\noindent{ \textbf{Proof:}}} {$\blacksquare$\vskip \belowdisplayskip}

\newenvironment{prevproof}[2]{\noindent {\em {Proof of {#1}~\ref{#2}:}}}{$\blacksquare$\vskip \belowdisplayskip}

\def\H{\mathcal{H}}

\def\E{\mathcal{E}}

\def\M{\mathcal{M}}

\def\I{\mathcal{I}}
\def\L{\mathcal{L}}

\def\A{{\mathcal{A}}}

\def\X{\mathcal{X}}
\def\Y{\mathcal{Y}}

\begin{document}

\title{Computing Equilibria in Anonymous Games}

\author{Constantinos Daskalakis~~~~~~~~~~~~~~~~~ Christos Papadimitriou \thanks{The authors were supported through NSF grant CCF - 0635319, a gift from Yahoo! Research and a 
MICRO grant.}\\
University of California, Berkeley\\ Computer Science \\ \{costis, christos\}@cs.berkeley.edu\\
}

\maketitle
\thispagestyle{empty}

\begin{abstract}
We present efficient approximation algorithms for finding Nash equilibria in anonymous games, that 
is, games in which the playersÕ utilities, though different, do not differentiate between other players. Our 
results pertain to such games with many players but few strategies. We show that any such game has an 
approximate pure Nash equilibrium, computable in polynomial time, with approximation $O(s^2 \lambda)$, where 
s is the number of strategies and $\lambda$ is the Lipschitz constant of the utilities. Finally, we show that there is 
a PTAS for finding an $\epsilon$-approximate Nash equilibrium when the number of strategies is two.
\end{abstract}


\section{Introduction}

Will you come to FOCS?  This decision depends on many factors, but one of them is {\em how many other theoreticians will come}.  Now, whether each of them will come depends in a small way on what you will do, and hence this aspect of the decision to go to FOCS is game-theoretic --- and in fact of a particular specialized sort explored in this paper:   Each player has a small number of strategies (in this example, two), and the utility of each player depends on her/his own decision, as well as on how many other players will choose each of these strategies.  It is crucial that the utilities do {\em not} depend on the identity of the players making these choices (that is, we ignore here your interest in whether your friends will come).  Such games are called {\em anonymous games}, and in this paper we give two polynomial algorithms for computing approximate equilibria in these games.  In fact, our algorithms work in a generalized framework:  The players can be divided into a few {\em types} (e.g., colleagues, students, big shots, etc.), and your utility depends on how many of the players {\em of each type} choose each of the strategies.  

Notice that this is a much more general framework than that of symmetric games (where all players are identical); each player can have a very individual way of evaluating the situation, and her/his utility can depend on the choices of others in an individual arbitrary way; in particular, there may be no monotonicity:  For example, a player may prefer a mob with 1000 attendees, mostly students, to a tiny workshop of 20, while a medium-sized conference of 200 may be more attractive than either; a second player may order these in the exact opposite  way.  Anonymous games comprise a broad and well studied class of games (see e.g. \cite{Blonski1, Blonski2, kalai} for recent work on this subject by economists) which are of special interest to the Algorithmic Game Theory community, as they capture important aspects of auctions and markets, as well as of Internet congestion.

Our interest lies in computing {\em Nash equilibria} in such games.  The problem of computing Nash equilibria in a game was recently shown to be PPAD-complete in general \cite{DGP}, even for the two-player case \cite{CD}.  Since that negative result, the research effort in this area was, quite predictably, directed towards two goals:  (1) computing approximate equilibria (mixed strategy profiles from which no player has incentive {\em more than $\epsilon$} to defect), and (2) exploring the algorithmic properties of special cases of games.  The approximation front has been especially fertile, with several positive and negative results shown recently \cite{LMM,CDT,KPS,DMP1,DMP2,FNS}.  

What is known about special cases of the Nash equilibrium problem?  Several important cases are now known to be generic; these include, beyond the aforementioned  2-player games, win-lose games (games with 0-1 utilities) \cite{AKV}, and several kinds of succinctly representable games such as graphical games \cite{DGP} and anonymous games (actually, the even more specialized {\em symmetric} games \cite{GaleKuhnTucker}).  For anonymous games, the genericity argument goes as follows:  Any game can be made anonymous by expanding the strategy space so that each player first chooses an identity (and is punished is s/he fails to choose her/his own) and then a strategy; it is easy to see that, in this expanded strategy space,  the utilities can be rendered in the anonymous fashion.   Note, however, that this manoeuvre requires a large strategy space; in contrast, for other succinct games such as the graphical ones, genericity persists even when the number of strategies is two \cite{GP}.  Are anonymous games easier when the number of strategies is fixed?  We shall see that this is indeed the case.

How about tractable special cases?  Here there is a relative poverty of results.  The zero-sum two-player case is, of course, well known \cite{vN}.  It was generalized in \cite{KT} to low-rank games (the matrix $A+B$ is not quite zero, but has fixed rank), a case in which a PTAS for the Nash equilibrium problem is possible.   It was also known that symmetric games with (about logarithmically) few strategies per player can be solved exactly in polynomial time by a reduction to the theory of  real closed fields \cite{PR}.  For congestion games we can find in polynomial time a pure Nash equilibrium if the game is a symmetric network congestion game \cite{FPT}, and an approximate pure Nash equilibrium if the congestion game is symmetric (but not necessarily network) and the utilities are somehow ``continuous''  \cite{CS}.  Finally, in \cite{Mi} Milchtaich showed that anonymous congestion games in graphs consisting of parallel links (equivalently, anonymous games in which the utility of a player, for each choice made by the player, is a nondecreasing function of the number of players who have chosen the same strategy) have pure Nash equilibria which can be computed in polynomial time by a natural greedy algorithm.

{\em In this paper we prove two positive approximation results for anonymous games with a fixed number $s$ of strategies.}  Our first result states that any such game has a {\em pure} Nash equilibrium that is $\epsilon$-approximate, where $\epsilon$  is bounded from above by a function of the form $f(s)\lambda$.  Here $\lambda$ is the {\em Lipschitz constant} of the utility functions (a measure of continuity of the utility functions of the players, assumed to be such that for any partitions $x$ and $y$ of the players into the $s$ strategies, $|u(x)-u(y)| \leq \lambda||x-y||_1$).  To get a sense of scale for $\lambda$ note that the arguments of $u$ range from $0$ to $n$ and so, if $u$  were a linear function in the range $[0,1]$, $\lambda$ would be at most $1\over n$.  $f(s)$ is a quadratic function of the number of strategies.  That $\epsilon$ cannot be smaller than $\lambda$ is easy to see (the matching pennies problem provides an easy example); the results of \cite{CS} for congestion games show a similar dependence on $\lambda$ (what they call ``the bounded jump property'').  We conjecture that the dependence on $s$ can be improved to $s\lambda$.  Our proof uses Brouwer's fixed point theorem on an interpolation of the (discrete) best-response function to identify a simplex of pure strategy profiles and from that produce, by a geometric argument, a pure strategy profile that is $\epsilon$-approximate, with $\epsilon$ bounded as above.

Our second result is a PTAS for the case of two strategies.  The main idea is to round the mixed strategies of the players to some nearby multiple of $\epsilon$; then each such quantized mixed strategy can be considered a pure strategy, and, with finitely many --in particular $O(1/\epsilon)$-- pure strategies, an anonymous game can be solved exhaustively in polynomial time in $n$, the number of players.  The only problem is, why should the expected utilities before and after the quantization be close?  Here we rely on a probabilistic lemma (Theorem \ref{theorem:main thm}) that may be of much more general interest:  Given $n$ Bernoulli random variables with probabilities $p_1, \ldots, p_n$, there is a way to round the probabilities to multiples of $1/ k$, for any $k$, so that the distribution of the sum of these $n$ variables is affected only by an additive $O\left(1/ \sqrt{k}\right)$ in total variational distance (no dependence on $n$).  This implies that the expected utilities of the quantized version are within an additive $\pm O\left({1/ \sqrt{k}}\right)$ of the original ones, and an $O\left(n^{1/ \epsilon^2}\right)$ PTAS for two-strategy anonymous games is immediate.  We feel that a more sophisticated proof of the same kind can establish a similar result for multinomial distributions, thus extending our PTAS to anonymous games with any fixed number of strategies.

\subsection{Definitions and Notation}
An {\em anonymous game} $G=(n,s,\{u^p_i\})$ consists of a set $[n]=\{1,\ldots,n\}$ of $n\geq 2$ of players, a set $[s]=\{1,\ldots,s\}$ of $s\geq 2$ strategies, and a set of $ns$ utility functions, where $u^p_i$ with $p\in [n]$ and $i\in[s]$  is the utility of player $p$ when she plays strategy $i$, a function mapping the set of partitions $\Pi^s_{n-1}=\{(x_1,\ldots,x_s): x_i \in N_0 \hbox{{\rm~for all~}} i\in[s], \sum_{i=1}^s x_i = n-1\}$ to the interval $[0,1]$ \footnote{In the literature on Nash approximation utilities are usually normalized this way so that the approximation error is additive.}.  Our working assumptions are that $n$ is large and $s$ is fixed; notice that, in this case, anonymous games are {\em succinctly representable} \cite{PR}, in the sense that their representation requires specifying $O(n^{s+1})$ numbers, as opposed to the $ns^n$ numbers required for general games (arguably, succinct games are the only multiplayer games that are computationally meaningful, see \cite{PR} for an extensive discussion of this point).  For our approximate pure Nash equilibrium result we shall also be assuming that the utility functions are continuous, in the following sense:  There is a real $\lambda>0$, presumably very small, such that  $|u^p_i(x)-u^p_i(y)| \leq \lambda\cdot ||x-y||_1$ for every $p\in[n], i\in[s],$ and $x,y\in\Pi^s_{n-1}$.  This continuity concept is similar to the ``bounded jump'' assumption of \cite{CS}.  The convex hull of the set $\Pi^s_{n-1}$ will be denoted by $\Delta^s_{n-1}= \{(x_1,\ldots,x_s): x_i \geq 0,~ i=1,\ldots,s, \sum_{i=1}^s x_i = n-1\}$.  

A {\em pure strategy profile} in such a game is a mapping $S$ from $[n]$ to $[s]$.  A pure strategy profile $S$ is an {\em $\epsilon$-approximate pure Nash equilibrium}, where $\epsilon > 0$, if, for all $p\in[n]$, $u^p_{S(p)}(x[S,p])+\epsilon  \geq u^p_i(x[S,p])$ for all $i\in[s]$, where $x[S,p]\in \Pi^s_{n-1}$ is the partition $(x_1,\ldots,x_s)$ such that $x_i$ is the number of players $q\in[n]-\{p\}$ such that $S(q)=i$.  

A  {\em mixed strategy profile} is a set of $n$ distributions $\delta_p, p\in[n]$, over $[s]$.  A mixed strategy profile is an {\em $\epsilon$-approximate mixed Nash equilibrium} if, for all $p\in[n]$ and $j\in[s]$, $E_{\delta_1,\ldots,\delta_n}u^p_i(x)+\epsilon  \geq E_{\delta_1,\ldots,\delta_n}u^p_j(x)$ where, for the purposes of the expectation, $i$ is drawn from $[s]$ according to $\delta_p$ and $x$ is drawn from $\Pi^s_{n-1}$ by drawing $n-1$ random samples from $[s]$ independently according to the distributions $\delta_q, q\neq p$ and forming the induced partition.

Anonymous games can be extended to ones in which there is also a finite number of {\em types} of players, and utilities depend on how each type is partitioned into strategies; all our algorithms, being exhaustive, can be easily generalized to this framework, with the number of types multiplying the exponent.

\section{Approximate Pure Equilibria}
In this section we prove the following result:
\begin{theorem}\label{th:pure appx}
In any anonymous game with $s$ strategies and Lipschitz constant $\lambda$ there is an $\epsilon$-approximate pure Nash equilibrium, where $\epsilon = O(s^2)\lambda$.
\end{theorem}

\begin{proof}  We first define a function $\phi$ from $\Pi^s_{n-1}$ to itself:  For any $x\in \Pi^s_{n-1}$, $\phi(x)$ is defined to be $(y_1,\ldots,y_s)\in \Pi^s_{n-1}$ such that, for all $i\in[s]$, $y_i$ is the number of all those players $p$ among $\{1,\ldots, n-1\}$ (notice that player $n$ is excluded) such that, for all $j<i$, $u^p_i[x] > u^p_j[x]$, and, for all $j>i$, $u^p_i[x] \geq u^p_j[x]$.  In other words, $\phi(x)$ is the partition induced among the first $n-1$ players by their best response to $x$, where ties are broken lexicographically.   

We next interpolate $\phi$ to obtain a continuous function $\hat \phi$ from $\Delta^s_{n-1}$ to itself as follows:  For each $x\in \Delta^s_{n-1}$ let us break $x$ into its integer and fractional parts $x=x^I+x^F$, where $x^I\in \Pi^s_{n-1}$ and $0\leq x^F_i\leq 1$ for all $i = 1,\ldots,s-1$.  Let $c[x^I]$ be the {\em cell} of $x^I$, the set of all $x'\in \Pi^s_{n-1}$ such that, for all $i\leq s-1$,  $x'_i= x^I_i$ or $x'_i= x^I_i+1$.  Then it is clear that $x$ can be written as a convex combination of the elements of $c[x^I]$:  $x=\sum_{x^j\in c[x^I]} \alpha_jx^j$.  We define $\hat \phi(x)$ to be  $\sum_{x^j\in c[x^I]} \alpha_j\phi(x^j)$.  

It is possible to define the interpolation at each point $x \in \Delta^s_{n-1}$ in a consistent way so that the resulting $\hat \phi(x)$ is a continuous  function from the compact set $\Delta^s_{n-1}$ to itself, and so, by Brouwer's Theorem, it must have a fixed point, that is, a point $x^*\in \Delta^s_{n-1}$ such that $\hat\phi(x^*)=x^*$.  That is,
$$\sum_{x^j\in c[x^{*I}]} \alpha_j\phi(x^j)=x^*=\sum_{x^j\in c[x^{*I}]} \alpha_jx^j.\eqno{(1)}$$
By Carath\' eodory's lemma, equation (1) can be expressed as the sum of only $s$ of the $\phi(x^j)$'s 
$x^*= \sum_{j=1}^s \gamma_j \phi(x^j),$  and it is easy to see that it can be rewritten as
$$x^*=\phi(x^1)+\sum _{j=2}^{s} \gamma_j (\phi(x^j)-\phi(x^1)),\eqno{(2)}$$
for some $\gamma_j\geq 0$ with $\sum_j\gamma_j=1$.  

Recall that, in order to prove the theorem, we need to exhibit $\epsilon$-approximate  pure strategy profile.  If $x^*$ were an integer point, then we would be almost done (modulo the $n$-th player, of whom we take care last), and $x^*$ itself (actually, the strategy profile suggested by the partition $\phi(x^*)$) would be essentially a pure Nash equilibrium, because of the equation $x^*=\phi(x^*)$.  But in general $x^*$ will be fractional, and the various $\phi(x^j)$'s will be very far from $x^*$ (except that they happen to have $x^*$ in their convex hull).  Our plan is to show that $x^1$ (a vertex in the cell of $x^*$) is an approximate pure Nash equilibrium (again, considered as a pure strategy profile and forgetting for a moment the $n$-th player).   

The term $\phi(x^1)$ in equation (2) can be seen as a pure strategy profile $P$:  Each of the $n-1$ players chooses the strategy that is her/his best response to $x^1$.  Therefore, in this strategy profile everybody would be happy if everybody else played according to $x^1$.  The problem is, of course, that $\phi(x^1)$ can be very far from $x^1$.  We shall next  use equation (2) to ``move it'' close to $x^1$ (more precisely, close to $x^*$ which we know is $2s$-close in $L_1$ distance to $x^1$)  without changing the utilities much.   Looking at one of the other terms of (2), $(\phi(x^j)-\phi(x^1))$, we can think of it as the act of switching certain players from their best response to $x^1$ to their best response to $x^j$.  The crucial observation is that, {\em since $x^1$ and $x^j$ are at most $2s$ apart in $L_1$ distance (they both belong to the same cell), the change in utility for the switching players would be at most $4s\lambda$.}  

So, equation (2) suggests that a strategy profile close to $x^*$ can be obtained from $P$ by combining these $s-1$ flows, with little harm in utility for all players involved.  The problem is how to combine them so that the right individual players are switched (the situation is akin to integer multicommodity flow).  We write each flow $(\phi(x^j)-\phi(x^1))$ as the sum of $s^2$ terms of the form $f^j_{i,i'}$, signifying the number of individual players moved from strategy $i$ to strategy $i'$. We know that, for each such nonzero flow, there is a set of players $S^j_{i,i'}\subseteq[n]$ which can be moved with only $4s\lambda$ loss in utility. The union over $j$ of the $(s-1)$ sets $S^j_{i,i'}$ is denoted by $S_{i,i'}$ and the union over $i'$ of the sets $S_{i,i'}$ by $S_i$. The following lemma can be proved by an application of Hall's Theorem. 


\begin{lemma}
There exist disjoint subsets $T_{i,i'} \subseteq S_{i,i'}$, $i'=1,\ldots,s$, such that, for all $i'=1,\ldots,s$, $|T_{i,i'}| = \lfloor \sum_{j=2}^s \gamma_j f^j_{i,i'} \rfloor$.
\end{lemma}

\begin{proof}
Let us consider the bipartite graph with vertex set $S_i \sqcup [s]$ and an edge from a player $p\in S_i$ to a strategy $i' \in [s]$ if, for some $j$, $p \in S^j_{i,i'}$. To establish the result it is enough to show that there exists a generalized matching of players to pure strategies ---in which every player is matched to at most one pure strategy, so that, for all $i'$, strategy $i'$ is matched with at least $R_{i'}:=\lfloor \sum_{j=2}^s \gamma_j f^j_{i,i'} \rfloor$ players. 

By Hall's theorem, such a matching exists if every set of strategies $\mathcal{J} \subseteq [s]$ ``knows'' at least $\sum_{i' \in \mathcal{J}}{R_{i'}}$ players. Observe that, for all $j$, the family of sets $\{S^j_{i,i'}\}_{i'}$ are disjoint. Hence
$$\sum_{i' \in \mathcal{J}}{f^j_{i,i'}} \le \left|\bigcup_{i' \in \mathcal{J}}{S^j_{i,i'}} \right| \le \left|\bigcup_{i' \in \mathcal{J}}{S_{i,i'}} \right| = |\Gamma(\mathcal{J})|,$$
where $\Gamma(\mathcal{J})$ represents the neighborhood in $S_i$ of the pure strategies of the set $\mathcal{J}$. From the above equation it follows that
\begin{align*}
&\sum_j \gamma_j\sum_{i' \in \mathcal{J}}{f^j_{i,i'}} \le \sum_j \gamma_j |\Gamma(\mathcal{J})| \le |\Gamma(\mathcal{J})|\\&\Rightarrow \sum_{i' \in \mathcal{J}} \sum_j \gamma_j{f^j_{i,i'}}  \le |\Gamma(\mathcal{J})| \Rightarrow \sum_{i' \in \mathcal{J}} R_{i'}  \le |\Gamma(\mathcal{J})|,
\end{align*}
which completes the proof.
\end{proof}


Thus, by moving the players in $T_{i,i'}$ from $i$ to $i'$, for all pairs of $i$, $i'$, we obtain from strategy profile $P$ a new strategy profile $\tilde P$ in which each player's strategy is within $4s\lambda$ from their response to $x^1$, and such that the corresponding partition $\tilde x^*$ is, by equation (2) and the roundings in the lemma, at most $2s^2$ away from $x^*$, and hence at most $2s$ more away from $x^1$; let's call the distance bound $D=2s^2+2s$.  Since, for all players ---except for the last of course, $\tilde P$ is an $4s\lambda$-approximate response to $x^1$ and $\tilde P$ is within distance $D$ from $x^1$, it follows that $\tilde P$ is a $4(D+s)\lambda$-approximate best response to itself.

Finally, we turn to player $n$.  Adding the best response of player $n$ to $\tilde P$, and subtracting what player $i$ plays in $\tilde P$, we get a profile that is $2$ away, in $L_1$ distance, from $\tilde P$, thus making $\tilde P$ a $4(D+s+1)\lambda$-approximate Nash equilibrium and completing the proof.
\end{proof}

Since $\Pi^s_{n-1}$ has $O(n^s)$ points, and this is the length of the input, the algorithmic implication is immediate:

\begin{corollary}
In any anonymous game, an $\epsilon$-approximate pure Nash equilibrium, where $\epsilon$ is as in Theorem \ref{th:pure appx}, can be found in linear time.
\end{corollary}

\section{Approximate Mixed Nash Equilibria}
\subsection{A Probabilistic Lemma}
We start by a definition.  The {\em total variation distance} between two distributions $\mathbb{P}$ and $\mathbb{Q}$ supported on a finite set $\mathcal{A}$ is $$||\mathbb{P} - \mathbb{Q}|| = \frac{1}{2} \sum_{\alpha \in \A}{\left|\mathbb{P}(\alpha)-\mathbb{Q}(\alpha)\right|}.$$

\begin{theorem} \label{theorem:main thm}
Let $\{p_i \}_{i=1}^n$ be arbitrary probabilities, $p_i \in [0,1]$, for $i=1, \ldots, n,$ and let $\{X_i\}_{i=1}^n$ be independent indicator random variables, such that $X_i$ has expectation $\E[X_i]=p_i$, and let $k$ be a positive integer. Then there exists another set of probabilities $\{q_i\}_{i=1}^n$, $q_i \in [0,1]$, $i=1,\ldots,n$, which satisfy the following properties:

\begin{enumerate}

\item $||q_i - p_i|| = O(1/k)$, for all $i=1,\ldots,n$ \label{cond: closeness}

\item $q_{i}$ is an integer multiple of $\frac{1}{k}$, for all $i=1,\ldots,n$ \label{cond: integrality}

\item if $\{Y_i\}_{i=1}^n$ are independent indicator random variables such that $Y_i$ has expectation $\E[Y_i]=q_i$, then, 
$$\left|\left|\sum_{i}{X_i} - \sum_{i}{Y_i} \right|\right|= O(k^{-1/2}).$$ \label{cond: small variation distance}
and, moreover, for all $j=1,\ldots,n$,
$$\left|\left|\sum_{i\neq j}{X_i} - \sum_{i \neq j}{Y_i} \right|\right|= O(k^{-1/2}).$$
\end{enumerate}
\end{theorem}

\noindent From this, the main result of this section follows:

\begin{corollary} \label{cor:mixed strategies}
There is a PTAS for the mixed Nash equilibrium problem for two-strategy anonymous games.
\end{corollary}
\begin{proof} Let $(p_1,\ldots, p_n)$ be a mixed Nash equilibrium of the game. We claim that $(q_1,\ldots, q_n)$, where the $q_i$'s are the multiples of $1/k$ specified by Theorem \ref{theorem:main thm}, constitute a $O(1/\sqrt{k})$-approximate mixed Nash equilibrium. Indeed, for every player $i\in [n]$ and every strategy $m \in \{1,2\}$ for that player let us track the change in the expected utility of the player when the distribution over $\Pi^2_{n-1}$ defined by the $\{p_j\}_{j \neq i}$ is replaced by the distribution defined by the $\{q_j\}_{j \neq i}$. It is not hard to see that the absolute change is bounded by the total variation distance between the distributions of the $\sum_{j \neq i}X_j$ and the $\sum_{j \neq i}Y_j$ \footnote{Recall that all utilities have been normalized to take values in $[0,1]$.} where $X_j$, $Y_j$ are indicators corresponding to whether player $j$ plays strategy $2$ in the distribution defined by the $p_i$'s  and the $q_i$'s respectively, i.e. $\E[X_j]=p_j$ and $\E[Y_j] =q_j$. Hence, the change in utility is at most $O(1/\sqrt{k})$, which implies that the $q_i$'s constitute an $O(1/\sqrt{k})$-approximate Nash equilibrium of the game, modulo the following observation: with a trivial modification in the proof of Theorem \ref{theorem:main thm} we can ensure sure that, when switching from $p_i$'s to $q_i$'s, for every $i$, the support of $q_i$ is a subset of the support  of $p_i$. 

To compute a quantized approximate Nash equilibrium of the original game, we proceed to define a related $(k + 1)$-strategy game, where $k = O\left(\frac{1}{\epsilon^2}\right)$, and treat the problem as a pure Nash equilibrium problem. It is not hard to see that the latter is efficiently solvable if the number of strategies is a constant. The new 
game is defined as follows: the $i$-th pure strategy, $i = 0, \ldots, k$, corresponds to a player in the original game playing strategy 2 with probability $\frac{i}{k}$. Naturally, the payoffs resulting from a pure strategy profile in the new game are defined to be equal to the corresponding payoffs in the original game, by the translation of the pure strategy profile of the former into a mixed strategy profile of the latter. In particular, for any player $p$, we can compute its payoff given any strategy $i \in \{0,\ldots, k\}$ for that player and any partition $x \in \prod^{k+1}_{n-1}$ of the other players into $k + 1$ strategies, in time $n^{O(1/\epsilon^2)}$ overall, by a straightforward dynamic programming algorithm, see for example \cite{Pap}. The remaining details are omitted.
\end{proof}

\noindent {\bf Remark:} Note that it is crucial for the proof of Corollary \ref{cor:mixed strategies} that the bound on the total variation distance between the $\sum_i{X_i}$ and the $\sum_i{Y_i}$  in the statement of the Theorem \ref{theorem:main thm} does not depend on the number $n$ of random variables which are being rounded, but only on the accuracy $\frac{1}{k}$ of the rounding. Because of this requirement, several simple methods of rounding are easily seen to fail:

\begin{itemize}
\item {\em Rounding to the Closest Multiple of $1/k$}: An easy counterexample for this method arises when $p_i:=\frac{1}{n}$, for all $i$. In this case, the trivial rounding would make $q_i := 0$, for all $i$, and the total variation distance between the $\sum_i X_i$ and the $\sum_i Y_i$ would become arbitrarily close to $1-\frac{1}{e}$, as $n$ goes to infinity.

\item {\em Randomized Rounding}: An argument employing the probabilistic method could start by independently rounding each $p_i$ to some random $q_i$ which is an integer multiple of $\frac{1}{k}$ in such a way that $\E[q_i] = p_i $. This seems promising since, by independence, for any $\ell=0,\ldots,n$, the random variable $\Pr[\sum_i Y_i = \ell]$, which is a function of the $q_i$'s, has the correct expectation, i.e. $\E[\Pr[\sum_i Y_i] = \ell] = \Pr[\sum_i X_i =\ell]$. The trouble is that the expectation of the random variable $\Pr[\sum_i Y_i = \ell]$ is very small: less than $1$ for all $\ell$ and, in fact, in the order of of $O(1/n)$ for many terms. Moreover, the function itself comprises of sums of products on the random variables $q_i$, in fact exponentially many terms for some values of $\ell$. Concentration seems to require $k$ which scales polynomially in $n$.
\end{itemize}

\noindent {\bf Proof Technique:} We follow instead a completely different approach which aims at directly approximating the distribution of the $\sum_i{X_i}$. The intuition is the following: The distribution of the $\sum_i{X_i}$ should be close in total variation distance to a Poisson distribution of the same mean $\sum_i p_i$. Hence, it seems that, if we define $q_i$'s ---which are multiples of $\frac{1}{k}$--- in such a way that the means $\sum_i p_i$ and $\sum_i q_i$ are close, then the distribution of the $\sum_i {Y_i}$ should be close in total variation distance to the same Poisson distribution and hence to the distribution of the $\sum_i{X_i}$ by triangle inequality.

There are several complications, of course, the main one being that the distribution of the $\sum_i{X_i}$ can be well approximated by a Poisson distribution of the same mean only when the $p_i$'s are relatively small. When the $p_i$'s take arbitrary values in $[0,1]$ and $n$ scales, the Poisson distribution can be very far from the distribution of the $\sum_iX_i$.  In fact, we wouldn't expect that the Poisson distribution can approximate the distribution of arbitrary sums of indicators since its mean and variance are the same. To counter this we resort to a special kind of distributions, called {\em translated Poisson distributions}, which are Poisson distributions appropriately shifted on their domain. An arbitrary sum of indicators can be now approximated as follows: a Poisson distribution is defined with mean --- and, hence, variance --- equal to the variance of the sum of the indicators; then the distribution is appropriately shifted on its domain so that its new mean coincides with the mean of the sum of the indicators being approximated.

The translated Poisson approximation will outperform the Poisson approximation for intermediate values of the $p_i$'s, while the Poisson approximation will remain better near the boundaries, i.e. for values of $p_i$ close to $0$ or $1$. Even for the intermediate region of values for the $p_i$'s, the translated Poisson approximation is not sufficient since it only succeeds when the number of the indicators being summed over is relatively large, compared to the minimum expectation. A different argument is required when this is not the case. Our bounding technique has to interleave these considerations in a very delicate fashion to achieve the approximation result. At a high level, we treat separately the $X_i$'s with small, medium or large expectation; in particular, for some $\alpha \in (0,1)$ to be fixed later, we define the following subintervals of $[0,1]$:

\begin{enumerate}
\item $\L(k):=\left[0, \frac{\lfloor k^{\alpha}\rfloor}{k}\right)$:
interval of {\em small expectations};

\item $\M_1(k):=\left[\frac{\lfloor k^{\alpha}\rfloor}{k},
\frac{k/2}{k}\right)$: first interval of {\em medium expectations};

\item $\M_2(k):=\left[\frac{k/2}{k}, 1-\frac{\lfloor
k^{\alpha}\rfloor}{k}\right)$: second interval of {\em medium expectations};

\item $\H(k):=\left[1-\frac{\lfloor k^{\alpha}\rfloor}{k},1\right]$:
interval of {\em high expectations}.
\end{enumerate}

\noindent Denoting $\L^*(k):=\{i~|~\E[X_i]\in \L(k)\}$, we
establish (Lemma \ref{lem:total variation distance small expectations}) that
$$\left|\left|\sum_{i\in \L^*(k)}{X_i} - \sum_{i\in \L^*(k)}{Y_i}\right| \right| = O(k^{-1/2})$$
and similarly for $\M_1(k)$ (Lemma \ref{lem:total variation distance first interval medium expectations}). Symmetric arguments (setting $X'_i=1-X_i$ and $Y'_i=1-Y_i$) imply the same bounds for the intervals $\M_2(k)$ and $\H(k)$. Therefore, an application of the coupling lemma implies that
$$\left| \left| \sum_{i}{X_i} - \sum_{i}{Y_i} \right| \right| = O(k^{-1/2}),$$
which concludes the proof. The details of the proof are postponed to Section \ref{sec: main proof}. The proof for the partial sums $\sum_{i\neq j}{X_i}$ and $\sum_{i \neq j}{Y_i}$ follows easily from the analysis of Section \ref{sec: main proof} and its details are skipped for this extended abstract. The next section provides the required background on Poisson approximations.

\subsection{Poisson Approximations} \label{sec:Poissonappx}

The following theorem is classical in the theory of Poisson approximations.

\begin{theorem}[\cite{BarbourEtAl:book}] \label{lem:Poisson approximation}
Let $J_1,\ldots,J_n$ be a sequence of independent random indicators with $\E[J_i]=p_i$. Then
$$\left|\left|\sum_{i=1}^nJ_i - Poisson\left(\sum_{i=1}^np_i\right)\right|\right| \le \frac{\sum_{i=1}^np_i^2}{\sum_{i=1}^np_i}.$$
\end{theorem}

\noindent As discussed in the previous section, the above bound is sharp when the indicators have small expectations, but loose when the indicators are arbitrary. The following approximation bound becomes sharp when the previous is not. But first let us formally define the translated Poisson distribution.

\begin{definition}[\cite{Rollin:translatedPoissonApproximations}] We say that an integer random variable $Y$ has a {\em translated Poisson distribution} with paremeters $\mu$ and $\sigma^2$ and write
$$\L(Y)=TP(\mu,\sigma^2)$$
if $\L(Y - \lfloor \mu-\sigma^2\rfloor) = Poisson(\sigma^2+ \{\mu-\sigma^2\})$, where $\{\mu-\sigma^2\}$ represents the fractional part of $\mu-\sigma^2$.
\end{definition}\text{}

\noindent Theorem \ref{lem:translated Poisson approximation} provides an approximation result for the translated Poisson distribution using Stein's method.

\begin{theorem}[\cite{Rollin:translatedPoissonApproximations}] \label{lem:translated Poisson approximation}
Let $J_1,\ldots,J_n$ be a sequence of independent random indicators with $\E[J_i]=p_i$. Then
$$\left|\left|\sum_{i=1}^nJ_i - TP(\mu, \sigma^2)\right|\right| \le \frac{\sqrt{\sum_{i=1}^np_i^3(1-p_i)}+2}{\sum_{i=1}^np_i(1-p_i)},$$
where $\mu=\sum_{i=1}^np_i$ and $\sigma^2 = \sum_{i=1}^np_i(1-p_i)$.
\end{theorem}

\noindent Lemmas \ref{lem:variation distance between Poisson distributions} and \ref{lem: variation distance between translated Poisson distributions} provide respectively bounds for the total variation distance between two Poisson distributions and two translated Poisson distributions with different parameters. The proof of \ref{lem:variation distance between Poisson distributions} is postponed to the appendix, while the proof of \ref{lem: variation distance between translated Poisson distributions} is provided in \cite{BarbourLindvall}.

\begin{lemma} \label{lem:variation distance between Poisson distributions}
Let  $\lambda_1, \lambda_2 \in \mathbb{R_+} \setminus \{0\}$ . Then
\begin{align*}&\left|\left|Poisson(\lambda_1)-Poisson(\lambda_2)\right|\right| \\ &~~~~~~~~~~~~~~~~~~~~~~~~~~~~~~\le e^{|\lambda_1-\lambda_2|}-e^{-|\lambda_1-\lambda_2|}.\end{align*}
\end{lemma}

\begin{lemma}[\cite{BarbourLindvall}] \label{lem: variation distance between translated Poisson distributions}
Let  $\mu_1, \mu_2 \in \mathbb{R}$ and $\sigma_1^2, \sigma_2^2 \in \mathbb{R}_+ \setminus \{0\}$ be such that $\lfloor \mu_1-\sigma_1^2 \rfloor \le \lfloor \mu_2-\sigma_2^2 \rfloor$. Then
\begin{align*}&\left|\left|TP(\mu_1,\sigma_1^2)-TP(\mu_2,\sigma_2^2)\right|\right|\\&~~~~~~~~~~~~~~~~~~~~~ \le \frac{|\mu_1-\mu_2|}{\sigma_1}+\frac{|\sigma_1^2-\sigma_2^2|+1}{\sigma_1^2}.\end{align*}
\end{lemma}

\subsection{Proof Theorem \ref{theorem:main thm}} \label{sec: main proof}

In this section we complete the proof of Theorem \ref{theorem:main thm}. As argued above, it is enough to round the random variables $\{X_i\}_{i\in \L^*(k)}$ into random variables $\{Y_i\}_{i\in \L^*(k)}$ so that the total variation distance between the random variables $\X_{\L}:=\sum_{i\in \L^*(k)}{X_i}$ and $\Y_{\L}:=\sum_{i\in \L^*(k)}{Y_i}$ is small and similarly for the subinterval $\M_1(k)$. 

Our rounding will have different objective in the two regions. When rounding the $X_i$'s with $i \in \L^*(k)$ we aim to approximate the mean of $\X_{\L}$ as tightly as possible. On the other hand, when rounding the $\X_i$'s with $i \in \M^*_1(k)$, we give up on approximating the mean very tightly in order to also approximate well the variance. The details of the rounding follow.\ \
 
\medskip \noindent {\bf Some notation first:}
Let us partition the interval $[0,1/2]$ into $\lceil k/2 \rceil$ subintervals $I_0, I_1, \ldots, I_{\lfloor k/2 \rfloor}$ where
$$I_0 = \left[0, \frac{1}{k}\right), I_1 = \left[\frac{1}{k}, \frac{2}{k}\right), \ldots, I_{\lfloor k/2 \rfloor} = \left[\frac{{\lfloor k/2 \rfloor}}{k},1/2\right].$$

\noindent The intervals $I_0,\ldots,I_{\lfloor k/2 \rfloor}$ define the partition
of $\L^*(k) \cup \M_1^*(k)$ into the subsets $I^*_0, I^*_1, \ldots,
I^*_{\lfloor k/2 \rfloor}$, where
$$I^*_j = \{ i~ |~ \E[X_i] \in I_j\}, j=0,1,\ldots,{\lfloor k/2 \rfloor}.$$

\noindent For all $j \in \{0,\ldots,{\lfloor k/2 \rfloor}\}$ with $I^*_j \neq
\emptyset$, let $I^*_j = \{j_1,j_2,\ldots, j_{n_j}\}$ and, for all
$i\in \{1,\ldots,n_j\}$, let

$$p^j_{i}:=\E[X_{j_i}] \text{ and }  \delta^j_i:=p^j_i-\frac{j}{k}.$$

\noindent We proceed to define the ``rounding'' of the $X_i$'s into the $Y_i$'s in the intervals $\L(k)$ and $\M_1(k)$ separately.\\
 
\subsubsection*{Interval $\L(k):=\left[0, \frac{\lfloor k^{\alpha}\rfloor}{k}\right)$ of small expectations.}

\noindent Observe first that $\L(k) \equiv I_0 \cup \ldots \cup
I_{\lfloor k^{\alpha}\rfloor-1}$ and define the corresponding subset of
the indices $\L^*(k):=I^*_0\cup \ldots \cup I^*_{\lfloor
k^{\alpha}\rfloor-1}$. We define the $Y_i$, $i \in \L^*(k)$, via the
following iterative procedure. Our ultimate goal is to round the
$X_i$'s into $Y_i$'s appropriately so that the sum of the
expectations of the $X_i$'s and of the $Y_i$'s are as close as
possible. The rounding procedure is as follows.
\begin{enumerate}
\item[i.] $\epsilon_0:=0;$

\item[ii.] for $j := 0 \text{ to }\lfloor k^{\alpha}\rfloor-1$

\begin{enumerate}
\item $S_j := \epsilon_j + \sum_{i=1}^{n_j}{\delta^j_i};$

\item $m_j:= \left\lfloor \frac{S_j}{k^{-1}}\right\rfloor$; $\epsilon_{j+1}  := S_j - m_j \cdot \frac{1}{k}$;
\item[] \{{\em assertion: $m_j \le n_j$ - see justification next}\}

\item set $q^j_i:= \frac{j+1}{k}$ for $i=1,...,m_j$ and $q^j_i:= \frac{j}{k}$ for $i=m_j+1,...,n_j$;

\item for all $i \in \{1,...,n_j\}$, let $Y_{j_i}$ be a $\{0,1\}$-random variable with expectation $q^j_i$;
\end{enumerate}
\item[iii.] Suppose that the random variables $Y_i$, $i \in
\L^*(k)$, are mutually independent.
\end{enumerate}\text{}

\noindent It is easy to see that, for all $j \in \{ 0,\ldots, \lfloor k^{\alpha} \rfloor-1\}$, $\epsilon_j < \frac{1}{k}$; this follows immediately from the description of the procedure, in particular Steps i and ii(b). This further implies that $m_j \le n_j$, for all $j$, since at Step ii(b) we have
$$\frac{S_j}{k^{-1}} = \frac{\epsilon_j + \sum_{i=1}^{n_j}{\delta^j_i}}{k^{-1}}< \frac{k^{-1} + n_j k^{-1}}{k^{-1}}=1+n_j.$$
Hence, the assertion following Step ii(b) is satisfied. Finally, note that, for all $j$,
\begin{align*}
\sum_{i=1}^{n_j}{q^j_i}&=m_j \frac{j+1}{k}+(n_j-m_j)\frac{j}{k}\\&=n_j\frac{j}{k}+m_j\frac{1}{k}=n_j\frac{j}{k}+S_j - \epsilon_{j+1}\\&=n_j\frac{j}{k}+ \sum_{i=1}^{n_j}{\delta^j_i} + \epsilon_j - \epsilon_{j+1} = \sum_{i=1}^{n_j}{p^j_i}+ \epsilon_j - \epsilon_{j+1}.
\end{align*}
Therefore,
$$\sum_{j=0}^{\lfloor k^{\alpha}\rfloor-1}{\sum_{i=1}^{n_j}{q^j_i}}=\sum_{j=0}^{\lfloor k^{\alpha}\rfloor-1}{\sum_{i=1}^{n_j}{p^j_i}}+\epsilon_0 - \epsilon_{\lfloor k^{\alpha}\rfloor},$$
which implies
\begin{lemma}\label{lemma:small expectations - expectation of sum is close}
$|\sum_{i \in \L^*(k)}\E[Y_i]-\sum_{i \in \L^*(k)}\E[X_i]| \le
\frac{1}{k}.$
\end{lemma}\text{}

\noindent The following lemma characterizes the total variation distance between $\sum_{i \in \L^*(k)}{X_i}$ and $\sum_{i \in \L^*(k)}{Y_i}$.

\begin{lemma} \label{lem:total variation distance small expectations}
$\left|\left|\sum_{i \in \L^*(k)}{X_i}-\sum_{i \in \L^*(k)}{Y_i}\right|\right| \le \frac{3}{k^{1-\alpha}}$.
\end{lemma}

\begin{proof}
By lemma \ref{lem:Poisson approximation} we have that

\begin{align*}
\left|\left|\sum_{i \in \L^*(k)}{X_i} - Poisson\left(\sum_{i \in \L^*(k)}p_i\right)\right|\right| \le \frac{\sum_{i \in \L^*(k)}p_i^2}{\sum_{i \in \L^*(k)}p_i}
\end{align*}
and
\begin{align*}
\left|\left|\sum_{i \in \L^*(k)}{Y_i} - Poisson\left(\sum_{i \in \L^*(k)}q_i\right)\right|\right| \le \frac{\sum_{i \in \L^*(k)}q_i^2}{\sum_{i \in \L^*(k)}q_i},
\end{align*}
where $p_i := \E[X_i]$ and $q_i := \E[Y_i]$ for all $i$. The following lemma is proven in the appendix.

\begin{lemma}\label{lem:optimization:1}
For any $u>0$ and any set $\{p_i\}_{i \in \I}$, where $p_i \in [0,u]$, for all $i\in \I$,
$$\frac{\sum_{i \in \I}p_i^2}{\sum_{i \in \I}p_i} \le u.$$
\end{lemma}

\noindent Using Lemmas \ref{lem:optimization:1}, \ref{lemma:small expectations - expectation of sum is close} and \ref{lem:variation distance between Poisson distributions} and the triangle inequality we get that

\begin{align*}\left|\left|\sum_{i \in \L^*(k)}{X_i}-\sum_{i \in \L^*(k)}{Y_i}\right|\right| &\le \frac{2}{k^{1-\alpha}}+(e^{\frac{1}{k}}-e^{-\frac{1}{k}})\\&\le \frac{3}{k^{1-\alpha}},\end{align*}
where we used the fact that $e^{\frac{1}{k}}-e^{-\frac{1}{k}} \le \frac{3}{k} \le \frac{1}{k^{1-\alpha}}$, for sufficiently large $k$.
\end{proof}

\subsubsection*{Interval $\M_1(k):=\left[\frac{\lfloor k^{\alpha}\rfloor}{k}, \frac{k/2}{k}\right)$: medium expectations.}
\noindent Observe first that $\M_1(k) \equiv I_{\lfloor k^{\alpha}\rfloor}
\cup \ldots \cup I_{\lfloor k/2 \rfloor}$ and define the
corresponding subset of the indices $\M_1^*(k):=I^*_{\lfloor
k^{\alpha}\rfloor} \cup \ldots \cup I^*_{\lfloor k/2 \rfloor}$. We
define the $Y_i$, $i \in \M_1^*(k)$, via the following procedure
which is slightly different than the one we used for the set of
indices $\L^*(k)$. Our goal here is to approximate well both the
mean and the variance of the sum $\sum_{i \in \M_1^*(k)}X_i$. In
fact, we will give up on approximating the mean as tightly as
possible, which we did above, in order achieve a good
approximation of the variance. The rounding procedure is as
follows.

\begin{enumerate}

\item[] for $j := \lfloor k^{\alpha}\rfloor \text{ to }\lfloor \frac{k}{2} \rfloor$

\begin{enumerate}
\item $S_j := \sum_{i=1}^{n_j}{\delta^j_i};$

\item $m_j:= \left\lfloor \frac{S_j}{k^{-1}}\right\rfloor$;


\item set $q^j_i:= \frac{j+1}{k}$ for $i=1,...,m_j$ and $q^j_i:=
\frac{j}{k}$ for $i=m_j+1,...,n_j$;

\item for all $i \in \{1,...,n_j\}$, let $Y_{j_i}$ be a
$\{0,1\}$-random variable with expectation $q^j_i$;
\end{enumerate}
\item[] Suppose that the random variables $Y_i$, $i \in \M_1^*(k)$,
are mutually independent.
\end{enumerate}\text{}

\noindent Lemma \ref{lemma:medium interval variances} characterizes the quality of the rounding
procedure in terms of mean and variance. Defining $\zeta_j:=\sum_{i \in I^*_j}\E[X_i]-\sum_{i \in
I^*_j}\E[Y_i],$ we have
\begin{lemma} \label{lemma:medium interval variances}
For all $j\in \{\lfloor k^{\alpha}\rfloor,\ldots,\lfloor \frac{k}{2} \rfloor\}$
\begin{enumerate}
\item[\em (a)]
$\zeta_j=\sum_{i=1}^{n_j}\delta^j_i-m_j\frac{1}{k}.$

\item[\em (b)] $0\le \zeta_j\le \frac{1}{k}.$

\item[\em (c)] $\sum_{i \in I^*_j}\text{\em Var}[X_i]=n_j\frac{j}{k}\left(1-\frac{j}{k}\right)+\left(1-\frac{2j}{k}\right)\sum_{i=1}^{n_j}{\delta^j_i} -
\sum_{i=1}^{n_j}{(\delta^j_i)^2}$

\item[\em (d)] $\sum_{i \in I^*_j}\text{\em Var}[Y_i]=n_j\frac{j}{k}\left(1-\frac{j}{k}\right)+m_j\frac{1}{k}\left(1-\frac{2j+1}{k}\right)$

\item[\em (e)] $\sum_{i \in I^*_j}\text{\em Var}[X_i]-\sum_{i \in
I^*_j}\text{\em
Var}[Y_i]=\left(1-\frac{2j}{k}\right)\zeta_j+\left(m_j\frac{1}{k^2}-\sum_{i=1}^{n_j}{(\delta^j_i)^2}
\right).$
\end{enumerate}
\end{lemma}\text{}\\

\noindent The following lemma bounds the total variation distance between the random variables $\sum_{i \in \M_1^*(k)}{X_i}$ and $\sum_{i \in \M_1^*(k)}{Y_i}$.

\begin{lemma} \label{lem:total variation distance first interval medium expectations}
$\left|\left|\sum_{i \in \M_1^*(k)}{X_i}-\sum_{i \in \M_1^*(k)}{Y_i}\right|\right| \le O\left(k^{-\frac{\alpha+\beta-1}{2}}\right)+O(k^{-\alpha})+O(k^{-\frac{1}{2}})+O(k^{-(1-\beta)}).$
\end{lemma}

\begin{proof}
We distinguish two cases for the size of $\M_1^*(k)$. For some $\beta \in (0,1)$ such that $\alpha+\beta>1$, let us distinguish two possibilities for the size of $|\M_1^*(k)|$:
\begin{enumerate}
\item[a.] $|\M_1^*(k)| \le k^{\beta}$

\item[b.] $|\M_1^*(k)| > k^{\beta}$
\end{enumerate}
Let us treat each interval separately in the following lemmas.

\begin{lemma}\label{lem: medle interval small cardinality}
If $|\M_1^*(k)| \le k^{\beta}$ then
$$\left|\left|\sum_{i \in \M_1^*(k)}{X_i}-\sum_{i \in \M_1^*(k)}{Y_i}\right|\right|\le \frac{1}{k^{1-\beta}}.$$
\end{lemma}

\begin{proof}
The proof follows from the coupling lemma and an easy coupling argument. The details are postponed to the appendix.
\end{proof}

\begin{lemma}\label{lem: middle interval big cardinality}
If $|\M_1^*(k)| > k^{\beta}$ then
\begin{align*}&\left|\left|\sum_{i \in \M_1^*(k)}{X_i}-\sum_{i \in \M_1^*(k)}{Y_i}\right|\right| \\ &~~~~~~~~~~~~~~\le O\left(k^{-\frac{\alpha+\beta-1}{2}}\right)+O(k^{-\alpha})+O(k^{-\frac{1}{2}}).\end{align*}
\end{lemma}
\begin{proof}
By lemma \ref{lem:translated Poisson approximation} we have that

\begin{align*}
&\left|\left|\sum_{i \in \M_1^*(k)}{X_i} - TP\left(\mu_1,\sigma_1^2\right)\right|\right|\\&~~~~~~~~~~~~~~~~~~~~~ \le  \frac{\sqrt{\sum_{i \in \M_1^*(k)}p_i^3(1-p_i)}+2}{\sum_{i \in \M_1^*(k)}p_i(1-p_i)}
\end{align*}
and
\begin{align*}
&\left|\left|\sum_{i \in \M_1^*(k)}{Y_i} - TP\left(\mu_2,\sigma_2^2\right)\right|\right|\\&~~~~~~~~~~~~~~~~~~~~~  \le  \frac{\sqrt{\sum_{i \in \M_1^*(k)}q_i^3(1-q_i)}+2}{\sum_{i \in \M_1^*(k)}q_i(1-q_i)}
\end{align*}
where $\mu_1=\sum_{i \in \M_1^*(k)}p_i$, $\mu_2=\sum_{i \in \M_1^*(k)}q_i$, $\sigma_1^2 = \sum_{i \in \M_1^*(k)}p_i(1-p_i)$, $\sigma_2^2 = \sum_{i \in \M_1^*(k)}q_i(1-q_i)$ and $p_i := \E[X_i]$, $q_i := \E[Y_i]$ for all $i$. The following lemma is proven in the appendix.

\begin{lemma}\label{lem:optimization:2}
For any $u \in (0,\frac{1}{2})$ and any set $\{p_i\}_{i \in \I}$, where $p_i \in [u,\frac{1}{2}]$, for all $i\in \I$,
$$\frac{\sqrt{\sum_{i \in \I}p_i^3(1-p_i)}}{\sum_{i \in \I}p_i(1-p_i)} \le \frac{(1+2u+4u^2-8u^3)}{\sqrt{16 |\I| u (1-u-4u^2+4u^3)}}$$
\end{lemma}\text{}\\

\noindent Applying the above lemma with $u=\frac{\lfloor k^{\alpha} \rfloor}{k}$ and $\I =  \M_1^*(k)$, where recall $|\M_1^*(k)|>k^{\beta}$ the above bound becomes
$$\frac{(1+2u+4u^2-8u^3)}{\sqrt{16 |\I| u (1-u-4u^2+4u^3)}}= O\left(k^{-\frac{\alpha+\beta-1}{2}}\right),$$
which implies
\begin{align}
\left|\left|\sum_{i \in \M_1^*(k)}{X_i} - TP\left(\mu_1,\sigma_1^2\right)\right|\right| =  O\left(k^{-\frac{\alpha+\beta-1}{2}}\right) \label{eq:dist1}
\end{align}
and
\begin{align}
\left|\left|\sum_{i \in \M_1^*(k)}{Y_i} - TP\left(\mu_2,\sigma_2^2\right)\right|\right| = O\left(k^{-\frac{\alpha+\beta-1}{2}}\right), \label{eq:dist2}
\end{align}
where we used that, for any set of values $\{p_i \in \M_1(k)\}_{i \in \M_1^*(k)}$,
\begin{align*}\frac{2}{\sum_{i \in \M_1^*(k)}p_i(1-p_i)} &\le \frac{2}{|\M_1^*(k)| \frac{\lfloor[k^{\alpha}\rfloor}{k}(1- \frac{\lfloor[k^{\alpha}\rfloor}{k})}\\&=O\left(k^{-({\alpha+\beta-1})}\right)
\end{align*}
and similarly for any set of values $\{q_i \in \M_1(k)\}_{i \in \M_1^*(k)}$.

All that remains to do is to bound the total variation distance between the distributions $TP(\mu_1,\sigma_1^2)$ and $TP(\mu_2,\sigma_2^2)$ for the parameters $\mu_1$, $\sigma_1^2$, $\mu_2$, $\sigma_2^2$ specified above. The following claim is proved in the appendix.

\begin{claim} \label{claim:total variation distance between translated Poissons} For the parameters specified above
$$\left|\left|TP(\mu_1,\sigma_1^2)-TP(\mu_2,\sigma_2^2)\right|\right| \le O(k^{-\alpha})+O(k^{-1/2}).$$
\end{claim}

\noindent Claim \ref{claim:total variation distance between translated Poissons} along with Equations \eqref{eq:dist1} and \eqref{eq:dist2} and the triangle inquality imply that

\begin{align*}&\left|\left|\sum_{i \in \M_1^*(k)}{X_i}-\sum_{i \in \M_1^*(k)}{Y_i}\right|\right|\\&~~~~~~~~~~~~~=O\left(k^{-\frac{\alpha+\beta-1}{2}}\right)+O(k^{-\alpha})+O(k^{-1/2}).\end{align*}
\end{proof}\end{proof}

\subsubsection*{Putting Everything Together.}

Suppose that the random variables $\{Y_i\}_i$ defined above are mutually independent. It follows that

\begin{align*}&\left|\left| \sum_{i}{X_i} - \sum_{i}Y_i\right| \right|= O\left( k^{-(1-\alpha)} \right) +O\left(k^{-\frac{\alpha+\beta-1}{2}}\right)\\&~~~~~~~~~~~~~~~~~~~~~~~+O(k^{-\alpha})+O(k^{-\frac{1}{2}})+O(k^{-(1-\beta)}).
\end{align*}
Setting $\alpha = \beta = \frac{3}{4}$ we get a total variation distance of $O\left( k^{-1/4} \right)$. A more delicate argument establishes an 
exponent of $-\frac{1}{2}$. 

\section{Open Problems}
Can our PTAS be extended to arbitrary fixed number of strategies?  We believe so.  A more sophisticated technique would subdivide, instead of the interval $[0,1]$ as our proof did, the $(s-1)$-dimensional simplex into domains in which multinomial (instead of binomial) distributions would be approximated in different ways, possibly using the techniques of \cite{Roos}.  This way of extending our result already seems to work for $s=3$, and we are hopeful that it will work for general fixed $s$.  

Can the quadratic, in $s$, approximation bound of our pure Nash equilibrium algorithm be improved to linear?  We believe so, and we conjecture that $s\lambda$ is a lower bound.

We hope that the ways of thinking about anonymous games introduced in this paper will eventually lead to algorithms for the practical solution of this important class of games. Moreover, a technique involving probability rounding similar to the one used here yields a quasi-polynomial time approximation scheme for finding a Nash equilibrium in general normal form games with a fixed number of strategies, as well as for large classes of graphical games of this sort  (work in progress).  Improving this to polynomial is another important open problem.


\medskip\noindent{\bf Acknowledgment:}  We want to thank Uri Feige for a helpful discussion.

\newpage
\onecolumn
\begin{appendix}
\section*{APPENDIX}
\section{Missing Proofs}

\begin{prevproof}{lemma}{lem:variation distance between Poisson distributions}
Without loss of generality assume that $0<\lambda_1 \le \lambda_2$ and denote $\delta = \lambda_2-\lambda_1$. For all $i\in\{0,1,\ldots\}$, denote 
$$p_i=e^{-\lambda_1}\frac{\lambda_1^i}{i!} \text{ ~~~~ and~~~   }q_i=e^{-\lambda_2}\frac{\lambda_2^i}{i!}.$$
Finally, define $\I^*=\{i:p_i\ge q_i\}$.\\

\noindent We have
\begin{align*}
\sum_{i \in \I^*}{|p_i - q_i|} = \sum_{i \in \I^*}{(p_i - q_i)}  &\le \sum_{i \in \I^*}{\frac{1}{i!}(e^{-\lambda_1}\lambda_1^i - e^{-\lambda_1-\delta}\lambda_1^i)}\\
&= \sum_{i \in \I^*}{\frac{1}{i!}e^{-\lambda_1}\lambda_1^i(1 - e^{-\delta})}\\
&\le (1 - e^{-\delta}) \sum_{i=0}^n{\frac{1}{i!}e^{-\lambda_1}\lambda_1^i}=1 - e^{-\delta}.
\end{align*}

\noindent On the other hand
\begin{align*}
\sum_{i \notin \I^*}{|p_i - q_i|} = \sum_{i \notin \I^*}{(q_i - p_i)}  &\le \sum_{i \notin \I^*}{\frac{1}{i!}(e^{-\lambda_1}(\lambda_1+\delta)^i - e^{-\lambda_1}\lambda_1^i)}\\
&= \sum_{i \notin \I^*}{\frac{1}{i!}e^{-\lambda_1}((\lambda_1+\delta)^i-\lambda_1^i)}\\
&\le \sum_{i=0}^n{\frac{1}{i!}e^{-\lambda_1}((\lambda_1+\delta)^i-\lambda_1^i)}\\
&= e^{\delta}\sum_{i=0}^n{\frac{1}{i!}e^{-(\lambda_1+\delta)}(\lambda_1+\delta)^i }- \sum_{i=0}^n{\frac{1}{i!}e^{-\lambda_1}\lambda_1^i}\\
&= e^{\delta}-1.
\end{align*}
Combining the above we get the result.
\end{prevproof}

\begin{prevproof}{lemma}{lem:optimization:1}
For all $i \in \I$ and any choice of values $p_j \in [0,u]$, $j\in \I\setminus \{i\}$, define the function
$$f(x)=\frac{x^2+A}{x+B}$$
where $A=\sum_{j \in \I\setminus \{i\}}p_j^2$ and $B=\sum_{j \in \I\setminus \{i\}}p_j$; observe that  $A\le u B$. \\

\noindent If $A=B=0$ then $f(x) = x$ so $f$ achieves its maximum at $x=u$.

\noindent If $A\neq0\neq B$, the derivative of $f$ is
$$f'(x)=\frac{-A + x (2 B + x)}{(B + x)^2}.$$
Denoting by $h(x)$ the numerator of the above expression, the derivative of $h$ is
$$h'(x)=2 B + 2 x >0, \forall x.$$
Therefore, $h$ is increasing which implies that $h(x)=0$ has at most one root and hence $f'(x)=0$ has at most one root since the denominator in the above expression for $f'(x)$ is always positive. Note that $f'(0) < 0$ whereas $f'(u)>0$. Therefore, there exists a unique $\rho \in (0,u)$ such that $f'(x)<0$, $\forall x\in(0,\rho)$, $f'(\rho)=0$ and $f'(x)>0$, $\forall x\in(\rho,u)$, i.e. $f$ is decreasing in $(0,\rho)$ and increasing in $(\rho,u)$. This implies that 
$$\max_{x\in[0,u]} f(x) = \max\{f(0), f(u)\}.$$
But $f(u)=\frac{u^2+A}{u+B} \ge \frac{A}{B} = f(0)$ since  $A\le u B$. Therefore, $\max_x f(x) = f(u)$.\\

\noindent From the above it follows that, independent of the values of  $\sum_{j \in \I\setminus \{i\}}p_i^2$ and $\sum_{j \in \I\setminus \{i\}}p_i$, $f(x)$ is maximized at $x=u$. Therefore, the expression $\frac{\sum_{i \in \I}p_i^2}{\sum_{i \in \I}p_i}$ is maximized when $p_i=u$ for all $i\in \I$. This implies
$$\frac{\sum_{i \in \I}p_i^2}{\sum_{i \in \I}p_i} \le \frac{|\I|u^2}{|\I|u}\le u.$$
\end{prevproof}

\begin{prevproof}{lemma}{lemma:medium interval variances}
For (a) we have
\begin{align*}
\zeta_j=\sum_{i \in I^*_j}\E[X_i] - \sum_{i \in I^*_j}\E[Y_i] &=
\sum_{i=1}^{n_j}p^j_i - \sum_{i=1}^{n_j}q^j_i \\&=
\sum_{i=1}^{n_j}\left(\frac{j}{k}+\delta^j_i\right)-\left(m_j
\frac{j+1}{k} +
(n_j-m_j)\frac{j}{k}\right)\\&=\sum_{i=1}^{n_j}\delta^j_i-m_j
\frac{1}{k}.
\end{align*}
To get (b) we notice that
\begin{align*}
\sum_{i=1}^{n_j}\delta^j_i - m_j\frac{1}{k}&=S_j- \lfloor S_jk \rfloor\frac{1}{k}\\
&=S_j-(S_jk-(S_jk-\lfloor S_jk \rfloor))\frac{1}{k}\\
&=(S_jk-\lfloor S_jk \rfloor))\frac{1}{k}.
\end{align*}

\noindent For (c), (d) and (e), observe that, for all $j_i \in I^*_j$,
$i\in\{1,\ldots,n_j\}$,

\begin{align*}
\text{Var}[X_{j_i}] = p^j_i (1-p^j_i) = p^j_i - (p^j_i)^2 &=
\E[X_{j_i}] - \left(\frac{j}{k}+\delta^j_i\right)^2\\
&=\frac{j}{k}+\delta^j_i - \left(\frac{j}{k}+\delta^j_i\right)^2
\end{align*}

\begin{align*}
\text{Var}[Y_{j_i}] = q^j_i - (q^j_i)^2 &=
\begin{cases}
\E[Y_{j_i}]-\frac{(j+1)^2}{k^2},~~~~~~~~~~i \in \{1,\ldots,m_j\}\\
\E[Y_{j_i}]-\frac{j^2}{k^2},~~~~~~~~~~~~~~~~i \in  \{m_j+1,\ldots,n_j\}
\end{cases}\\
&=\begin{cases}\frac{j+1}{k}-\frac{(j+1)^2}{k^2},~~~~~~~~~~~i \in \{1,\ldots,m_j\} \\
\frac{j}{k}-\frac{j^2}{k^2},~~~~~~~~~~~~~~~~~~~~~i \in
\{m_j+1,\ldots,n_j\}
\end{cases}
\end{align*}

\noindent Hence,
\begin{align*}
\sum_{i \in I^*_j}\text{\em Var}[X_i]= \sum_{i=1}^{n_j}\text{\em
Var}[X_{j_i}]&=\sum_{i=1}^{n_j}\E[X_{j_i}] - n_j\frac{j^2}{k^2}
-2\frac{j}{k}\sum_{i=1}^{n_j}{\delta^j_i} -
\sum_{i=1}^{n_j}{(\delta^j_i)^2}\\&=n_j\frac{j}{k} + \sum_{i=1}^{n_j}{\delta^j_i}- n_j\frac{j^2}{k^2}
-2\frac{j}{k}\sum_{i=1}^{n_j}{\delta^j_i} -
\sum_{i=1}^{n_j}{(\delta^j_i)^2}\\
&=n_j\frac{j}{k}\left(1-\frac{j}{k}\right)+\left(1-\frac{2j}{k}\right)\sum_{i=1}^{n_j}{\delta^j_i} -
\sum_{i=1}^{n_j}{(\delta^j_i)^2}.
\end{align*}

\begin{align*}
\sum_{i \in I^*_j}\text{\em Var}[Y_i]=\sum_{i=1}^{n_j}\text{\em
Var}[Y_{j_i}]= &\sum_{i=1}^{n_j}\E[Y_{j_i}]
-n_j\frac{j^2}{k^2}-m_j\frac{2j+1}{k^2}\\
&=n_j\frac{j}{k}+m_j\frac{1}{k}-n_j\frac{j^2}{k^2}-m_j\frac{2j+1}{k^2}\\
&=n_j\frac{j}{k}\left(1-\frac{j}{k}\right)+m_j\frac{1}{k}\left(1-\frac{2j+1}{k}\right).
\end{align*}
and
\begin{align*}
&\sum_{i \in I^*_j}\text{\em Var}[X_i]-\sum_{i \in I^*_j}\text{\em
Var}[Y_i]=\\
&~~~~~~~~~~= \left(n_j\frac{j}{k}\left(1-\frac{j}{k}\right)+\left(1-\frac{2j}{k}\right)\sum_{i=1}^{n_j}{\delta^j_i} -
\sum_{i=1}^{n_j}{(\delta^j_i)^2}\right)-\left(n_j\frac{j}{k}\left(1-\frac{j}{k}\right)+m_j\frac{1}{k}\left(1-\frac{2j+1}{k}\right)\right)\\
&~~~~~~~~~~=
\left(1-\frac{2j}{k}\right)\left( \sum_{i=1}^{n_j}{\delta^j_i} - m_j\frac{1}{k} \right)+\left(m_j\frac{1}{k^2}-\sum_{i=1}^{n_j}{(\delta^j_i)^2}
\right)\\
&~~~~~~~~~~=
\left(1-\frac{2j}{k}\right)\zeta_j+\left(m_j\frac{1}{k^2}-\sum_{i=1}^{n_j}{(\delta^j_i)^2}
\right).
\end{align*}
\end{prevproof}

\begin{prevproof}{Lemma}{lem: medle interval small cardinality}
The coupling lemma implies that for any joint distribution on $\{X_i\}_i \cup \{Y_i\}_i$ the following is satisfied
$$\left|\left|\sum_{i \in \M_1^*(k)}{X_i}-\sum_{i \in \M_1^*(k)}{Y_i}\right|\right| \le \Pr\left[ \sum_{i \in \M_1^*(k)}{X_i} \neq \sum_{i \in \M_1^*(k)}{Y_i}\right].$$
A union bound further implies
$$\Pr\left[ \sum_{i \in \M_1^*(k)}{X_i} \neq \sum_{i \in \M_1^*(k)}{Y_i}\right] \le \Pr\left[ \bigvee_{i \in \M_1^*(k)}{X_i \neq Y_i}\right] \le \sum_{i \in \M_1^*(k)} \Pr \left[ X_i \neq Y_i\right].$$
Hence for any joint distribution on $\{X_i\}_i \cup \{Y_i\}_i$ the following is satisfied
\begin{align}\left|\left| \sum_{i \in \M_1^*(k)} {X_i} - \sum_{i \in \M_1^*(k)}{Y_i}\right|\right|  \le \sum_{i \in \M_1^*(k)}{\Pr \left[ X_i \neq Y_i\right]}. \label{formula1}
\end{align}

\noindent Let us now choose a joint distribution on $\{X_i\}_i \cup \{Y_i\}_i$ in which, for all $i$, $X_i$ and $Y_i$ are coupled in such a way that
$$\Pr[X_i \neq Y_i] \le \frac{1}{k}.$$
This is easy to do since by construction $|p_i - q_i| \le \frac{1}{k}$, for all $i$. Plugging in Formula \eqref{formula1} the particular joint distribution just described yields
$$\left|\left| \sum_{i \in \M_1^*(k)} {X_i} - \sum_{i \in \M_1^*(k)}{Y_i}\right|\right|  \le \frac{|\M_1^*(k)|}{k} \le \frac{1}{k^{1-\beta}}.$$
\end{prevproof}

\begin{prevproof}{lemma}{lem:optimization:2}
For all $i \in \I$ and any choice of values $p_j \in [u,\frac{1}{2}]$, $j\in \I\setminus \{i\}$, define the function
$$f(x)=\frac{x^3(1-x)+A}{(x(1-x)+B)^2}, x \in \left[u,\frac{1}{2}\right]$$
where $A=\sum_{j \in \I\setminus \{i\}}p_j^3(1-p_j)$ and $B=\sum_{j \in \I\setminus \{i\}}p_j(1-p_j)$. 
\\

\noindent For the sake of the argument let us extend the range of $f$ to $[0,\frac{1}{2}]$. The derivative of $f$ is
$$f'(x)=\frac{A(-2 + 4 x) + x^2 (3 B + x - 4 B x - x^2)}{(B + x - x^2)^3},$$
where note that the denominator is positive for all $x\in [0,\frac{1}{2}]$. Denoting by $h(x)$ the numerator of the above expression, the derivative of $h$ is
$$h'(x)=4A + (1 - 2x)x^2 + 2x(x - x^2) + 6B x(1 - 2x)>0, \forall x \in \left[0,\frac{1}{2}\right].$$
Therefore, $h$ is increasing in $(0,\frac{1}{2})$ which implies that $h(x)=0$ has at most one root in $(0,\frac{1}{2})$ and hence $f'(x)=0$ has at most one root in $(0,\frac{1}{2})$ since the denominator in the above expression for $f'(x)$ is always positive. Note that $f'(0) < 0$ whereas $f'(\frac{1}{2})>0$. Therefore, there exists a unique $\rho \in (0,\frac{1}{2})$ such that $f'(x)<0$, $\forall x\in(0,\rho)$, $f'(\rho)=0$ and $f'(x)>0$, $\forall x\in(\rho,\frac{1}{2})$, i.e. $f$ is decreasing in $(0,\rho)$ and increasing in $(\rho,\frac{1}{2})$. This implies that 
$$\max_{x\in[u,1/2]} f(x) = \max\left\{f(u), f (1/2)\right\}.$$

\noindent Hence, the expression $\frac{\sum_{i \in \I}p_i^3(1-p_i)}{(\sum_{i \in \I}p_i(1-p_i))^2}$ is bounded by
$$\max_{m \in \{0,\ldots,|\I|\}}{\frac{\frac{m}{16}+(|\I|-m)u^3(1-u)}{(\frac{m}{4}+(|\I|-m)u(1-u))^2}}$$
which we further bound by
$$\max_{x \in [0,|\I|]}{g(x)}$$
where
$$g(x):=\frac{\frac{x}{16}+(|\I|-x)u^3(1-u)}{(\frac{x}{4}+(|\I|-x)u(1-u))^2}.$$
The derivative of $g$ is
$$g'(x)=\frac{4 |\I|u(1 + u - 6 u^2 + 12u^3 - 8u^4) - (1 - 2u - 16u^3 + 48 u^4 - 32 u^5)x}{(1 - 2 u)^{-1}(4 |\I| (1 - u)u + (1 - 2u)^2x)^3 },$$
where note that the denominator is positive for all $u\in[0,\frac{1}{2}]$ and the numerator is of the form $A(u)-B(u)x$ where $$A(u):=4 |\I|u(1 + u - 6 u^2 + 12u^3 - 8u^4) > 0, \forall u\in\left(0,\frac{1}{2}\right)$$ and $$B(u):=1 - 2u - 16u^3 + 48 u^4 - 32 u^5> 0, \forall u\in\left(0,\frac{1}{2}\right).$$ Hence, if we take $\zeta:=\frac{A(u)}{B(u)}$, $g'(x)>0$, $\forall x\in(0,\zeta)$, $g'(\zeta)=0$ and $g'(x)<0$, $\forall x\in(\zeta,+\infty)$, i.e. $g$ is increasing in $(0,\zeta)$ and decreasing in $(\zeta,+\infty)$. This implies that $g$ achieves its maximum at $x:=\zeta$. The maximum value itself is
$$g(\zeta)=\frac{(1+2u+4u^2-8u^3)^2}{16 |\I| u (1-u-4u^2+4u^3)}.$$
This concludes the proof.
\end{prevproof}

\begin{prevproof}{claim}{claim:total variation distance between translated Poissons}
From Lemma \ref{lem: variation distance between translated Poisson distributions} if follows that
$$\left|\left|TP(\mu_1,\sigma_1^2)-TP(\mu_2,\sigma_2^2)\right|\right| \le \max\left\{\frac{|\mu_1-\mu_2|}{\sigma_1}+\frac{|\sigma_1^2-\sigma_2^2|+1}{\sigma_1^2}, \frac{|\mu_1-\mu_2|}{\sigma_2}+\frac{|\sigma_1^2-\sigma_2^2|+1}{\sigma_2^2}\right\}.$$

\noindent Denoting $J=\{\lfloor k^{\alpha}\rfloor,\ldots, \lfloor k/2 \rfloor-1\}$, we have that

$$\mu_1-\mu_2=\sum_{i \in \M_1^*(k)}\E[X_i]-\sum_{i \in \M_1^*(k)}\E[Y_i] = \sum_{j \in J}{\sum_{i \in I^*_j}{(\E[X_i]-\E[Y_i])}= \sum_{j \in J}\zeta_j}.$$

\begin{align*}
\sigma_1^2 =  \sum_{j \in J}\sum_{i \in I^*_j}\text{\em Var}[X_i]&=\sum_{j \in J}\left(n_j\frac{j}{k}\left(1-\frac{j}{k}\right)+\left(1-\frac{2j}{k}\right)\sum_{i=1}^{n_j}{\delta^j_i} -
\sum_{i=1}^{n_j}{(\delta^j_i)^2}\right)\\
&\ge\sum_{j \in J}\left(n_j\frac{j}{k}\left(1-\frac{j}{k}\right)- n_j \frac{1}{k^2}\right)
\\&\ge\sum_{j \in J}\frac{n_j}{k^2} \left(j\left(k-j\right)- 1\right).
\end{align*}

\noindent Similarly
\begin{align*}
\sigma_2^2 =  \sum_{j \in J}\sum_{i \in I^*_j}\text{\em Var}[Y_i]&= \sum_{j \in J}\left(n_j\frac{j}{k}\left(1-\frac{j}{k}\right)+m_j\frac{1}{k}\left(1-\frac{2j+1}{k}\right)\right)\\
&\ge\sum_{j \in J}\left(n_j\frac{j}{k}\left(1-\frac{j}{k}\right)\right)
\\&\ge\sum_{j \in J}\frac{n_j}{k^2} j\left(k-j\right)
\\&\ge\sum_{j \in J}\frac{n_j}{k^2} \left(j\left(k-j\right)- 1\right).
\end{align*}

\noindent Finally,
\begin{align*}
\sigma_1^2-\sigma_2^2 = \sum_{j \in J} \left( \sum_{i \in I^*_j}  \text{\em Var}[X_i]-\sum_{i \in
I^*_j}\text{\em Var}[Y_i] \right)=\sum_{j \in J}\left(\left(1-\frac{2j}{k}\right)\zeta_j+\left(m_j\frac{1}{k^2}-\sum_{i=1}^{n_j}{(\delta^j_i)^2}
\right)\right),
\end{align*}
where observe that $\sigma_1^2-\sigma_2^2 \ge 0$, since
\begin{align*}
\left(1-\frac{2j}{k}\right)\zeta_j+\left(m_j\frac{1}{k^2}-\sum_{i=1}^{n_j}{(\delta^j_i)^2}
\right) &\ge \left(1-\frac{2j}{k}\right)\zeta_j+\left(m_j\frac{1}{k^2}-\frac{1}{k}\sum_{i=1}^{n_j}{\delta^j_i}
\right)\\
&= \left(1-\frac{2j}{k}\right)\zeta_j-\frac{1}{k}\zeta_j=\\
&= \left(1-\frac{2j+1}{k}\right)\zeta_j \ge 0.
\end{align*}

\noindent We proceed to bound each of the terms $ \frac{|\mu_1-\mu_2|}{\sigma_1}$, $\frac{|\sigma_1^2-\sigma_2^2|+1}{\sigma_1^2}$, $ \frac{|\mu_1-\mu_2|}{\sigma_2}$ and $\frac{|\sigma_1^2-\sigma_2^2|+1}{\sigma_2^2}$ separately. We have\\

\begin{align*}
\frac{|\mu_1-\mu_2|}{\sigma_1} = \frac{\mu_1-\mu_2}{\sigma_1} &\le \frac{\sum_{j \in J} \zeta_j }{\sqrt{\sum_{j \in J}\frac{n_j}{k^2} \left(j\left(k-j\right)- 1\right)}}\\
&\le \frac{k\sum_{j \in J} \zeta_j }{\sqrt{\sum_{j \in J: n_j \ge 1}\left(j\left(k-j\right)- 1\right)}}\\
&\le \frac{Z}{\sqrt{\sum_{j \in J: n_j \ge 1}\left(j\left(k-j\right)- 1\right)}}~~~~~~~~~~\text{where $Z= |\{j\in J| n_j \ge 1\}|$}\\
&= \frac{1}{\sqrt{\sum_{j \in J: n_j \ge 1}\left(\frac{j\left(k-j\right)- 1}{Z^2}\right)}}\\
&\le \frac{1}{\sqrt{\sum_{j=\lfloor k^{\alpha}\rfloor}^{\lfloor k^{\alpha}\rfloor+Z-1} \left(\frac{j\left(k-j\right)- 1}{Z^2}\right)}}\\
&\le \frac{1}{\sqrt{\min_{1\le Z \le k/2-\lfloor k^{\alpha} \rfloor}\left\{\sum_{j=\lfloor k^{\alpha}\rfloor}^{\lfloor k^{\alpha}\rfloor+Z-1} \left(\frac{j\left(k-j\right)- 1}{Z^2}\right)\right\}}}.
\end{align*}
Note that
\begin{align}
\sum_{j=j_1}^{j_2+Z-1} \left(j\left(k-j\right)- 1\right)&=\frac{1}{6} Z (-7 - 6 j_1^2 + 6 j_1(1 + 
        k - Z) + 3k(-1 + Z) + 3Z - 2Z^2). \label{eq:summation1}
\end{align}
For $j_1 = \lfloor k^{\alpha} \rfloor$, let us define
$$F(Z):=\frac{1}{Z^2}\sum_{j=j_1}^{j_2+Z-1} \left(j\left(k-j\right)- 1\right) \equiv \frac{1}{6Z}  (-7 - 6 j_1^2 + 6 j_1(1 + 
        k - Z) + 3k(-1 + Z) + 3Z - 2Z^2)$$

\noindent The derivative of $F$ is
$$F'(Z)=\frac{7 + 6 j_1^2 + 3 k - 6 j_1(1 + k) - 2 Z^2}{6 Z^2} = \frac{7 + 6 \lfloor k^{\alpha} \rfloor^2 + 3 k - 6\lfloor k^{\alpha} \rfloor(1 + k) - 2 Z^2}{6 Z^2} < 0, \text{ for large enough $k$}. $$
Hence $F$ is decreasing in $[1, k/2-\lfloor k^{\alpha} \rfloor]$ so it achieves its minimum at $Z=k/2-\lfloor k^{\alpha} \rfloor$. The minimum itself is
$$F(k/2-\lfloor k^{\alpha} \rfloor) = \frac{-14 + 6\lfloor k^{\alpha} \rfloor - 4 \lfloor k^{\alpha} \rfloor^2 - 3 k +4\lfloor k^{\alpha} \rfloor k + 2k^2}{6 k -12\lfloor k^{\alpha} \rfloor} = \Omega(k).$$
Hence,
\begin{align*}
\frac{|\mu_1-\mu_2|}{\sigma_1} = O(k^{-1/2}).
\end{align*}
Similarly, we get
\begin{align*}
\frac{|\mu_1-\mu_2|}{\sigma_2} = O(k^{-1/2}).
\end{align*}
It remains to bound the terms  $\frac{|\sigma_1^2-\sigma_2^2|+1}{\sigma_1^2}$ and $\frac{|\sigma_1^2-\sigma_2^2|+1}{\sigma_2^2}$. We have
\\

\begin{align*}
\frac{|\sigma_1^2-\sigma_2^2|+1}{\sigma_1^2}&\le \frac{1+\sum_{j \in J} ( \zeta_j+ m_j\frac{1}{k^2} )}{\sum_{j \in J}\frac{n_j}{k^2} \left(j\left(k-j\right)- 1\right)}\\
&\le \frac{1+k \sum_{j \in J} ( k \zeta_j)+ M}{\sum_{j \in J}n_j\left(j\left(k-j\right)- 1\right)}~~~~~~~~~\text{(where $M:=\sum_jm_j$)}\\
&\le \frac{1+k Z+ M}{\sum_{j \in J}n_j\left(j\left(k-j\right)- 1\right)}~~~~~~~~~\text{(where $Z= |\{j\in J| n_j \ge 1\}|$)}\\
&= \frac{k Z}{\sum_{j \in J}n_j\left(j\left(k-j\right)- 1\right)}+\frac{1 + M}{\sum_{j \in J}n_j\left(j\left(k-j\right)- 1\right)}\\
\end{align*}
The second term of the above expression is bounded as follows
\begin{align*}
\frac{1 + M}{\sum_{j \in J}n_j\left(j\left(k-j\right)- 1\right)} &\le \frac{1}{\sum_{j \in J}n_j\left(j\left(k-j\right)- 1\right)} + \frac{M}{\sum_{j \in J}n_j\left(j\left(k-j\right)- 1\right)}  \\
&\le \frac{1}{\sum_{j \in J}n_j\left(j\left(k-j\right)- 1\right)} + \frac{M}{\sum_{j \in J}n_j\left(j\left(k-j\right)- 1\right)}  \\
&\le \frac{1}{\lfloor k^{\alpha}\rfloor (k-\lfloor k^{\alpha}\rfloor)-1} + \frac{1}{\sum_{j \in J}\frac{m_j}{M}\left(\lfloor k^{\alpha}\rfloor (k-\lfloor k^{\alpha}\rfloor)- 1\right)}  \\
&\le \frac{2}{\lfloor k^{\alpha}\rfloor (k-\lfloor k^{\alpha}\rfloor)-1} = O(k^{-1-\alpha}).
\end{align*}
The first term is bounded as follows
\begin{align*}
\frac{k Z}{\sum_{j \in J}n_j\left(j\left(k-j\right)- 1\right)} &\le \frac{k Z}{\sum_{j \in J:n_j\ge1}\left(j\left(k-j\right)- 1\right)}  \\
&\le \frac{1}{\sum_{j=\lfloor k^{\alpha}\rfloor}^{\lfloor k^{\alpha}\rfloor+Z-1} \left(\frac{j\left(k-j\right)- 1}{kZ}\right)}\\
&\le \frac{1}{\min_{1\le Z \le k/2-\lfloor k^{\alpha} \rfloor}\left\{\sum_{j=\lfloor k^{\alpha}\rfloor}^{\lfloor k^{\alpha}\rfloor+Z-1} \left(\frac{j\left(k-j\right)- 1}{kZ}\right)\right\}}.
\end{align*}
From above
\begin{align}
\sum_{j=j_1}^{j_2+Z-1} \left(j\left(k-j\right)- 1\right)&=\frac{1}{6} Z (-7 - 6 j_1^2 + 6 j_1(1 + 
        k - Z) + 3k(-1 + Z) + 3Z - 2Z^2). \label{eq:summation1}
\end{align}
For $j_1 = \lfloor k^{\alpha} \rfloor$, let us define
$$G(Z):=\frac{1}{Zk}\sum_{j=j_1}^{j_2+Z-1} \left(j\left(k-j\right)- 1\right) \equiv \frac{1}{6k}  (-7 - 6 j_1^2 + 6 j_1(1 + 
        k - Z) + 3k(-1 + Z) + 3Z - 2Z^2)$$

\noindent The derivative of $G$ is
$$G'(Z)=\frac{3 - 6 j_1 + 3 k -4Z}{6 k} = \frac{3 - 6 \lfloor k^{\alpha} \rfloor + 3 k -4Z}{6 k}> 0, \text{ for large enough $k$ since $Z\le k/2$}. $$
Hence $G$ is increasing in $[1, k/2-\lfloor k^{\alpha} \rfloor]$ so it achieves its minimum at $Z=1$. The minimum itself is
$$G(1) = \frac{-6 - 6\lfloor k^{\alpha} \rfloor^2 + 6\lfloor k^{\alpha} \rfloor k }{6 k} = \Omega(k^{\alpha}).$$
Hence,
$$\frac{k Z}{\sum_{j \in J}n_j\left(j\left(k-j\right)- 1\right)} = O(k^{-\alpha}),$$
which together with the above implies
\begin{align*}
\frac{|\sigma_1^2-\sigma_2^2|+1}{\sigma_1^2} = O(k^{-\alpha})
\end{align*}
Similarly, we get
\begin{align*}
\frac{|\sigma_1^2-\sigma_2^2|+1}{\sigma_2^2} = O(k^{-\alpha}).
\end{align*}
Putting everything together we get that
$$\left|\left|TP(\mu_1,\sigma_1^2)-TP(\mu_2,\sigma_2^2)\right|\right| \le O(k^{-\alpha})+O(k^{-1/2}).$$
\end{prevproof}

\end{appendix}
\end{document}